\documentclass[letterpaper, 10 pt, conference]{ieeeconf}  

\IEEEoverridecommandlockouts                              
\usepackage{graphics} 
\usepackage{epsfig} 
\usepackage{mathptmx} 
\usepackage{times} 
\usepackage{amsmath} 
\usepackage{amssymb}  
\usepackage[T1]{fontenc}
\usepackage{newtxmath}
\usepackage{euscript}
\DeclareMathAlphabet{\mathpzc}{T1}{pzc}{m}{it}
\usepackage{mathtools}
\usepackage{blkarray, bigstrut} %

\newtheorem{theorem}{Theorem}
\newtheorem{definition}{Definition}
\newtheorem{remark}{Remark}
\newtheorem{assumption}{Assumption}
\newtheorem{lemma}{Lemma}

\usepackage{svg}
\usepackage{pdfpages}
\providecommand{\abs}[1]{\lvert#1\rvert}
\usepackage{caption}
\usepackage{subcaption}
\usepackage{cite}

\title{\LARGE \bf
A Communication Security Game on Switched Systems \\
for Autonomous Vehicle Platoons
}

\author{Guoxin Sun, Tansu Alpcan, Benjamin I. P. Rubinstein and Seyit Camtepe
\thanks{G. Sun, T. Alpcan and B. Rubinstein are with the
        University of Melbourne, Australia
        {\tt\small \{guoxins@student, tansu.alpcan@, brubinstein@\}.unimelb.edu.au}}%
\thanks{S. Camtepe is with the Distributed Systems Security Group, CSIRO Data61, Australia
        {\tt\small Seyit.Camtepe@data61.csiro.au}}%
}

\begin{document}

\maketitle
\thispagestyle{empty}
\pagestyle{empty}

\begin{abstract}

Vehicle-to-vehicle communication enables autonomous platoons to boost traffic efficiency and safety, while ensuring string stability with a constant spacing policy. However, communication-based controllers are susceptible to a range of cyber-attacks. In this paper, we propose a distributed attack mitigation defense framework with a dual-mode control system reconfiguration scheme to prevent a compromised platoon member from causing collisions via message falsification attacks. In particular, we model it as a switched system consisting of a communication-based cooperative controller and a sensor-based local controller and derive conditions to achieve global uniform exponential stability (GUES) as well as string stability in the sense of platoon operation. The switching decision comes from game-theoretic analysis of the attacker and the defender's interactions. In this framework, the attacker acts as a leader that chooses whether to engage in malicious activities and the defender decides which control system to deploy with the help of an anomaly detector. Imperfect detection reports associate the game with imperfect information. A dedicated state constraint further enhances safety against bounded but aggressive message modifications in which a bounded solution may still violate practical constraint e.g. vehicles nearly crashing. Our formulation uniquely combines switched systems with security games to strategically improve the safety of such autonomous vehicle systems.
\end{abstract}

\section{INTRODUCTION}
Connected and autonomous vehicles have emerged as an extensive and promising research area over the past two decades~\cite{qayyum2020securing}. As a closely related topic, vehicular platooning earns its reputation by providing driving/passenger comfort, improved energy efficiency, pollution reduction as well as increase of traffic throughput. The platooning concept involves a group of vehicles travelling in a tightly coupled manner from an origin to a destination as a single unit. A platoon member receives other vehicles' dynamics and maneuver-related information via a vehicle-to-vehicle (V2V) communication network to compute control commands accordingly and maintain platoon stability, i.e. to maintain a narrow inter-vehicle distance and relative velocity. 

However, such V2V communication implementations also expose novel attack vectors, which increase security vulnerabilities and highlight vehicle platoons as an appealing target for cyber-physical attacks. Adversaries could inject multiple falsified vehicle nodes into the platoon remotely, which allows them to publish carefully crafted beacon messages to gain the privilege of the road or to cause traffic congestion and even serious collisions~\cite{boeira2017effects}.  \textit{There is urgent need to address the safety risks caused by such communication-based attacks}.
This paper focuses on a longitudinal control system for vehicular platooning, characterised by two upper controllers, which provide individual vehicle stability and/or string stability. We propose
a novel dual-mode control system reconfiguration scheme that involves a switching process between the two controllers. To ensure stable switching, we firstly provide conditions on controller gains that is sufficient to guarantee global uniform exponential stability (GUES). Secondly, a minimum dwell time constraint for the string stable controller is then provided, which means this controller is required to be activated at least for this amount of the time to mitigate the defects brought by the other controller in terms of string stability. Thirdly, a security game with imperfect information is then constructed and solved to guide the switching process based on Nash equilibrium solutions, which creates a hybrid combination of game theory and control theory. This game models intelligent and stealthy attacks from rational adversaries, who attempts to bypass detection using the knowledge about the system. It takes the limitations of existing anomaly detection approaches into account while modeling the interactions between attacker and defender. 
Fourthly, a dedicated switching surface is also introduced to capture practical constraints. Lastly, the effectiveness of the proposed approach is shown with some numerical and simulation examples. 
The contributions of this paper include:
 \begin{itemize}
     \item We present a dual-mode control system reconfiguration scheme to mitigate communication-based attacks.
     along with a sufficient condition in terms of controller gains to ensure stability of the switched system.
     \item We provide a lower bound on the dwell time of the string-stable controller to ensure string stability of the switched system after an attack is mitigated.
     \item We develop a unique approach that uses game theory to guide a switched system for communication-based attack mitigation purposes. Our security game formulation captures imperfect detectors as a chance node in our security game structure, and takes detection errors (e.g., false alarms and misses) into account.
     \item The results are illustrated using sophisticated, system-level simulations.
 \end{itemize}
\vspace{-0.1cm}

The rest of the paper is organised as follows. 
Section~\ref{Background} formulates the considered platoon framework. Section~\ref{Attack Model} presents the attack model. The proposed control system reconfiguration scheme is discussed in Section~\ref{DefenseFramework} along with the derived stability conditions. Game theoretic analysis is performed in Section~\ref{Game}. Section~\ref{conclusion} outlines some concluding remarks and future work.

\subsection{Literature Review} \label{sec:litreview}

Sumra \emph{et al.}~\cite{sumra2015attacks} provide a comprehensive survey of the attacks on major security goals, i.e., confidentiality, integrity and availability. Malicious attackers can breach privacy by attempting an eavesdropping attack to steal and misuse confidential information \cite{wiedersheim2010privacy}. The use of vehicular botnets to attempt a denial-of-service (DoS) or distributed denial-of-service (DDoS) attack may cause serious network congestion or disruption \cite{zhang2020distributed, feng2020dynamic}. The attacker may disrupt the platoon operation by hijacking the sensors to conduct the replay attack with the aim to feed the control system with outdated signals generated previously by the system~\cite{merco2018replay}.

Several attempts have been made to detect communication-based attacks. A sliding mode observer is proposed for cyber-attack detection and estimation in the case of event-triggered communication~\cite{keijzer2019sliding}. 
There is also a growing body of literature that recognises the effectiveness of machine learning based intrusion detection system applied on vehicular platooning~\cite{dadras2018identification, yang2019tree, alotibi2020anomaly, sunstrategic}. Even though such studies aim to maximise their detection performance, inevitable false alarm and miss rates are problematic for real-world applications. Moreover, attack mitigation still remains as an open and active research area. 
An intelligent adversaries may also use their knowledge about system vulnerabilities to perform stealthy attacks that maximise their effect and minimise detection rate. 

There has been an increasing interest in cybersecurity analysis from a game-theoretic viewpoint \cite{alpcan2010network}. Extensive research has been carried out to examine problems of allocating limited defense resources (e.g. energy \cite{sedjelmaci2016lightweight}, manpower \cite{fang2016deploying}, communication bandwidth \cite{subba2018game1}) against adversaries in a network.  In addition, several studies~\cite{zohdy2012game,dextreit2013game,marden2015game,junhui2013power} have also applied game theory for the design of control systems. For instance, authors in\cite{dextreit2013game} model the interaction of the driver and the powertrain of an electric vehicle as a non-coorperative game and construct a controller based on a feedback Stackelberg equilibrium. Game theory is also applied to coordinate autonomous vehicles at an uncontrolled intersection to minimize the total delay~\cite{zohdy2012game}. The switched system concept provides a systems-oriented alternative pathway to mitigate attack effects and enhance safety in such adversarial environments. A specific class of switched systems is defined as a family of linear time-invariant systems whose parameters vary within a single finite set according to a switching signal or switching surface. There is large body of literature on observer design for switched systems with unknown inputs \cite{yang2017simultaneous, zammali2020interval, van2014hybrid} and attempts to stabilise the system under such situations \cite{lee2006optimal, yang2016finite, sanchez2019practical}. The present work builds on these existing literature and introduces a novel dual-mode control system reconfiguration scheme defined as a switched system consisting of a communication-based cooperative controller and a sensor-based local controller. This article presents a novel and systematic approach into game-theory-powered switching system as a local online defense strategy.

\section{BACKGROUND}\label{Background}

\subsection{Autonomous Vehicle Platoon Model}
We consider a hierarchical longitudinal control structure~\cite{rajamani2011vehicle}, with an upper level controller and a lower lever controller as shown in Fig.~\ref{fig:HighLevelControlStuct}. The upper level controller uses information on other vehicles acquired via sensors or communication and internal vehicle dynamics to compute the desired acceleration $a_{dir}$ for each vehicle. The lower level controller generates actuator inputs (e.g. throttle and/or brake commands) to track the desired acceleration.  
\begin{figure}[tp]
    \centering
    \includegraphics[width=\linewidth]{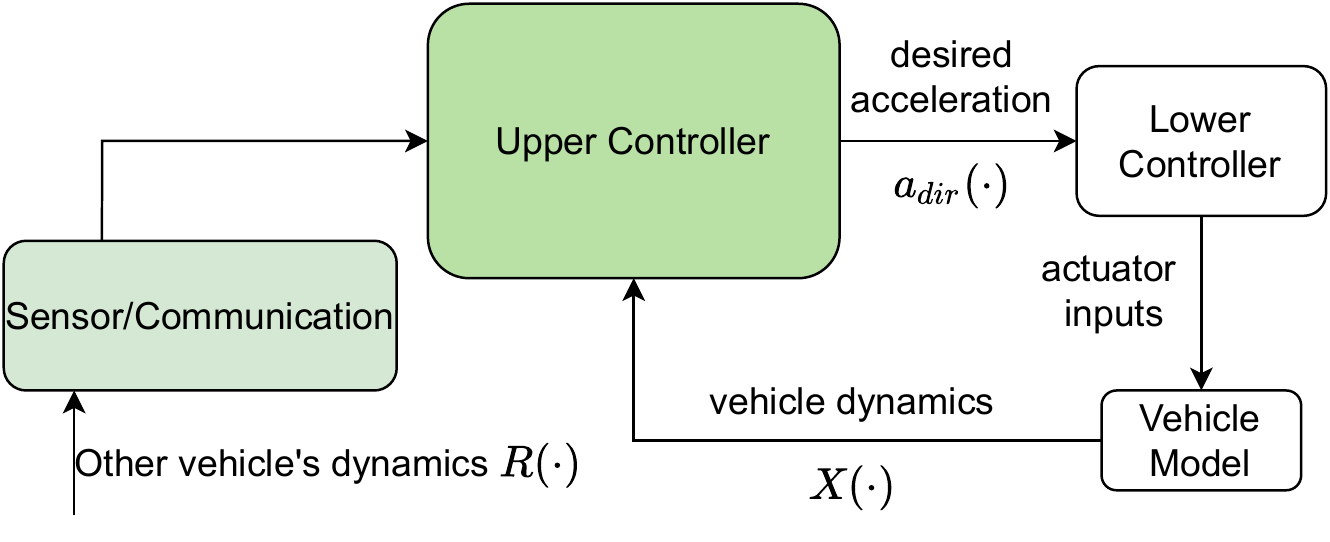}
    \caption{Structure of the longitudinal control system from the perspective of a vehicle in a platoon.}
    \label{fig:HighLevelControlStuct}
\end{figure}

In more detail, consider a platoon with $2\leq N< \infty$ vehicles. The $i^{th}$ vehicle's state is defined as $[x_i(t), v_i(t)]^T$, where $x_i(t)\in \mathbb{R}$ and $v_i(t)\in \mathbb{R}$ are the position and velocity of vehicle $i$. For the present work, we focus on the upper level controller in which each vehicle follows second-order dynamics
\begin{equation} \label{eq:UpperController}
\begin{split}
    \dot x_i(t) &= v_i(t),\\
    \dot v_i(t) &= u_i(t),
\end{split}
\end{equation}
where $u_i(t)$ is the control input (acceleration) to the system. To simplify the notation, time $t$ is omitted thereafter. The control policies considered in our defense framework can take either of the following two forms:\\

\noindent \textbf{1)} \textbf{Cooperative Adaptive Cruise Control (CACC) } \\
The control input in CACC is given by
\begin{equation}\label{CACC]}
    u_i = \sum_{j\in \mathpzc{N}_{i}} \alpha_{ij}(x_i-x_j+L_{ij}) +\sum_{j\in \mathpzc{N}_{i}}\beta_{ij} (v_i-v_j)+\sum_{j\in \mathpzc{N}_{i}} \gamma_{ij} a_j ,
\end{equation}
where the set $\mathpzc{N}_{i}$ contains vehicles that communicate with vehicle $i$ (i.e., vehicle  $i$'s \textit{neighbors}) and $a_j$ is the acceleration of vehicle $j$. Here $\alpha_{ij}\in \mathbb{R}$, $\beta_{ij}\in \mathbb{R}$ and $\gamma_{ij}\in \mathbb{R}$ are controller gains. In this way, the $i^{th}$ vehicle adjusts the desired acceleration in order to coordinate its speed with its \textit{neighbors} and to maintain the relative position of itself with its  around a desired (or targeted) inter-vehicle distance $L_{ij}$. 
We choose the desired distance $L_{ij} = L\Delta_{i,j}$, where
$\Delta_{i,j}\in \mathbb{N}$ is the number of vehicles (hops) between vehicle $i$ and $j$, and $L$ is constant and uniform for all vehicles of the platoon. 
Note that other vehicles' dynamics information, the position, speed and acceleration tuple $(x_j,v_j,a_j)$, is acquired via wireless communication through a V2V communication network, which can be implemented for example in the form of a 5G or vehicular ad-hoc network. \\

\noindent \textbf{2)} \textbf{Adaptive Cruise Control (ACC)}\\
In this controller, the control input is given by
\begin{equation}\label{ACC]}
    u_i =\sum_{j\in \mathpzc{N}_{i}} \alpha_{ij}(x_i-x_j+L_{ij}) +\sum_{j\in \mathpzc{N}_{i}}\beta_{ij} (v_i-v_j) ,
\end{equation}
where the set $j\in \mathpzc{N}_{i}$ contains the \textit{neighbors} of vehicle $i$ that are detectable by on-board range sensors. As above, $\alpha_{ij}\in \mathbb{R}$ and $\beta_{ij}\in \mathbb{R}$ are control gains. ACC control policy uses only the relative position and velocity as feedbacks in order to generate the desired acceleration to maintain a prefixed  inter-vehicle distance $L_{ij}$ and relative velocity. In comparison to CACC, other vehicles' dynamics information is obtained only by sensor measurements, which we assume in this paper to be reliable unlike communication messages.

In general, the double-integrator feedback system \eqref{eq:UpperController} for vehicle $i$ can be represented as 
\begin{equation}\label{eq:MatrixForm}
    \dot z = Az + B R,
\end{equation}
where $z = [x_i(t), v_i(t)]^T$ is the state vector
for vehicle $i$ (abusing notation for simplicity) and 
$$ R=  [[x_{j}-L_{ij}] [v_{j}] [a_j]]^T, \ \ \forall j\in \mathpzc{N}_{i},$$ 
is an external input vector that consists of other vehicle dynamics, where $[\cdot]$ represent a row vector of appropriate size. The matrix $A$ will have the following matrix forms respectively, depending on whether control policy CACC \eqref{CACC]} or ACC \eqref{ACC]} is in place
\begin{subequations}
 \begin{equation}
         \label{A:CACC}   
         A_{CACC} = \left[
\begin{array}{cc}
 0 & 1 \\
 k_1 & k_2 \\
\end{array}
\right]= \left[
\begin{array}{cc}
 0 & 1 \\
 \sum_{j\in \mathpzc{N}_{i}} \alpha_{ij} & \sum_{j\in \mathpzc{N}_{i}}\beta_{ij}  \\
\end{array}
\right],
\end{equation}
 \begin{equation}
         \label{A:ACC}   
        A_{ACC} = \left[
\begin{array}{cc}
 0 & 1 \\
 k_3 & k_4 \\
\end{array}
\right]= \left[
\begin{array}{cc}
 0 & 1 \\
 \sum_{j\in \mathpzc{N}_{i}} \alpha_{ij} & \sum_{j\in \mathpzc{N}_{i}}\beta_{ij}  \\
\end{array}
\right],
\end{equation}
\end{subequations}
where $k$'s consist of the corresponding combinations of $\alpha$ and $\beta$ parameters as in \eqref{CACC]} and \eqref{ACC]}. Similarly, the matrix $B$ takes the following forms

\begin{subequations}
 \begin{equation}
     B_{CACC} = 
     \begin{bmatrix}
        [0]  & [0] &  [0]\\
      [-\alpha_{ij}]  & [-\beta_{ij}] & [\gamma_{ij}] \\
     \end{bmatrix}, \ \ \forall j\in \mathpzc{N}_{i},
 \end{equation}
  \begin{equation}
     B_{ACC} = \begin{bmatrix}
     \begin{array}{ccc}
        [0]  &  [0] &  [0]\\
        [-\alpha_{ij}]  & [-\beta_{ij}]&[0]\\
     \end{array}
     \end{bmatrix},\ \ \forall j\in \mathpzc{N}_{i},
 \end{equation}
\end{subequations}
where $[\cdot]$ represents a row vector of appropriate size. \\[-0.75em]

\noindent\textbf{Assumptions.} To simplify our analysis, we adopt a specific CACC setup as in \cite{segata2014plexe}, which is based on the \textit{predecessor-leader following} information flow topology \cite{zheng2015stability}. In particular, each vehicle receives communicated position, velocity and acceleration information from only the platoon leader and its immediate proceeding vehicle, which equivalently means $\mathpzc{N}_{i} = \{1,\,i-1\}$ in \eqref{CACC]}. 
We also assume each vehicle in the platoon only equips a Radar sensor at the vehicle's front, which measures the position and velocity of its predecessor (i.e., $\mathpzc{N}_{i} = \{i-1\}$ in \eqref{ACC]}). The platoon leader is also assumed to be driven by a human driver who is not affected by communication attacks or sensor noise. Communication noise and sensor noise are ignored due to their low impact on system safety compared to intentional attacks.

\subsection{Basic Platoon Stability Analysis}\label{BasicStability}

In the platoon operation with constant spacing policy, two fundamental specifications, namely individual vehicle stability and string stability, must be satisfied to achieve long-term stable operation.
\begin{definition}[Individual Vehicle Stability \cite{rajamani2011vehicle}]
 Let's define the spacing error for the $i^{th}$ vehicle to be $\epsilon_i = x_i-x_{i-1}+L$. The system \eqref{eq:MatrixForm} is said to have individual vehicle stability if the following condition holds:
\begin{equation}\label{indStableCond}
     \ddot x_{i-1} \rightarrow 0 \; \Rightarrow \; \epsilon_i\rightarrow 0,
     \;\;  i = 2,\dots,N .
\end{equation}
\end{definition}

This definition essentially means a vehicle achieves asymptotic stability if the preceding vehicle operates at a constant velocity.  Since (\ref{eq:MatrixForm}) is a linear time-invariant (LTI) system, this stability can be achieved according to Lemma~\ref{individualStabilityLemma}.
\begin{lemma}\label{individualStabilityLemma}
If conditions \eqref{CACCinvCod} hold for the CACC control system, then the system achieves bounded-input-bounded-output (BIBO) stability, and therefore individual vehicle stability is guaranteed.  
\begin{subequations}\label{CACCinvCod}
 \begin{equation}
    k_{1}<0,
\end{equation} 
\begin{equation}
k_{2}\leq -2 \sqrt{-k_{1}}.
\end{equation} 
\end{subequations}
\end{lemma}
\begin{proof}
The proof is standard BIBO argument. An LTI continuous-time system with state representation ($A,B,C$) is BIBO stable if and only if the eigenvalues of $A$ are in the left-hand complex plane. In other words, $A$ is a Hurwitz matrix \cite{chen1984linear}.
\end{proof}

Similarly, the following conditions on control gain must hold to achieve individual vehicle stability for ACC control policy
\begin{subequations}\label{ACCinvCod}
 \begin{equation}
    k_{3}<0,
\end{equation} 
\begin{equation}
k_{4}\leq -2 \sqrt{-k_{3}}.
\end{equation} 
\end{subequations}
\begin{remark}
Even if each system is analytically proven to be stable, the bounded solution resulting from bounded disturbances may violate practical constraints that lead to collisions. More details will be discussed in later sections.
\end{remark}
If the preceding vehicle is not operating at constant velocity (i.e., accelerating or braking), the spacing error $\epsilon_i$
is expected to be nonzero. Therefore it is important to make sure spacing errors are guaranteed not to amplify as they propagate down to the tail of the platoon. 
\begin{definition}[String Stability \cite{rajamani2011vehicle}]\label{def:stringStablity}

Let's define the spacing error for the $i^{th}$ vehicle to be $\epsilon_i = x_i-x_{i-1}+L$. The system \eqref{eq:MatrixForm} is said to have string stability if the following condition holds:
\begin{equation}\label{stringStability}
     \|\epsilon_i\|_\infty\leq \|\epsilon_{i-1}\|_\infty,
     \;\;  i = 2,\dots,N .
\end{equation}
\end{definition}
\begin{theorem}[\cite{rajamani2011vehicle}]
Let the spacing errors of consecutive vehicles be related by the transfer function $\hat{H}(s) = \frac{\epsilon_i}{\epsilon_{i-1}}$. The string stability condition \eqref{stringStability} holds, if 
\begin{equation*}
    \|\hat{H}(s)\|_\infty\leq 1,
\end{equation*} 
\begin{equation*}
     \hat{h}(t) > 0,
\end{equation*}
 where $\hat{h}(t)$ is the impulse response of $\hat{H}(s)$. \end{theorem}
The string stability proofs for both systems are omitted  due to space limitation. In summary, the ACC controller can achieve individual vehicle stability via proper controller tuning but fails to ensure string stability under a constant spacing policy. Whereas the CACC controller achieves both stabilities when vehicle-to-vehicle communication is established.  

\section{ATTACK MODEL} \label{Attack Model}
We consider a particular type of communication attack, namely a \textit{message falsification} attack as described in \cite{qayyum2020securing}. By continuously monitoring the communication network, the adversary may change the content of received messages and subsequently insert them back into the network. The presence of this type of attack could cause instabilities to the vehicle platoon or even collisions.

Let $\mathcal{U}$ represent the set of vehicles that are affected by the attack. The state of an affected vehicle evolves as 
\begin{equation} \label{eq:CACC_withAttack}
\begin{split}
    \dot x_i(t) &= v_i(t),\\
    \dot v_i(t) &= u_i(t)+\xi(t),\ \ i\in \mathcal{U},
\end{split}
\end{equation}
where $\xi(t)$ is the intentional modifications on communicated messages. Unlike a noise term that only moderately degrades the system performance, an adversary could target specific platoon members and carry out stealthy, adaptive and aggressive attacks, which compromises platoon safety.
\begin{assumption}\label{def:CommAttacknotAffectSensor}
 Communication-related cyber-attacks do not alter the physical properties of individual vehicles, which means the integrity of the vehicle sensors and controllers are protected.
\end{assumption}

Assumption~\ref{def:CommAttacknotAffectSensor} clarifies our focus on communication-based attacks in this paper.
In comparison to attacks on actuators and sensors, \textit{message falsification} attacks do not directly alter the targeted physical systems but achieve malicious outcomes through modifications of system inputs. Although the ACC controller shares similar feedback structure (i.e., $A$ matrix structure) with CACC, they have different mechanisms to obtain other vehicle dynamics information (through on-board sensors). The immunity against \textit{message falsification} attack and the fact that physical system is still reliable under such attack mean the ACC controller is an appropriate  supplementary controller in such an adversarial environment. However, \textit{recall that ACC can not guarantee string stability, which limits its suitability for long-term deployment in a platoon}.


\section{CONTROLLER SWITCHING FOR ATTACK MITIGATION}\label{DefenseFramework}
 In this section, we  firstly present the overall structure of the  dual-mode control system reconfiguration scheme and find the conditions on controller gains that lead to GUES and string stability for the proposed system.

\begin{figure}[tb]
    \centering
    \includegraphics[width=0.9\columnwidth]{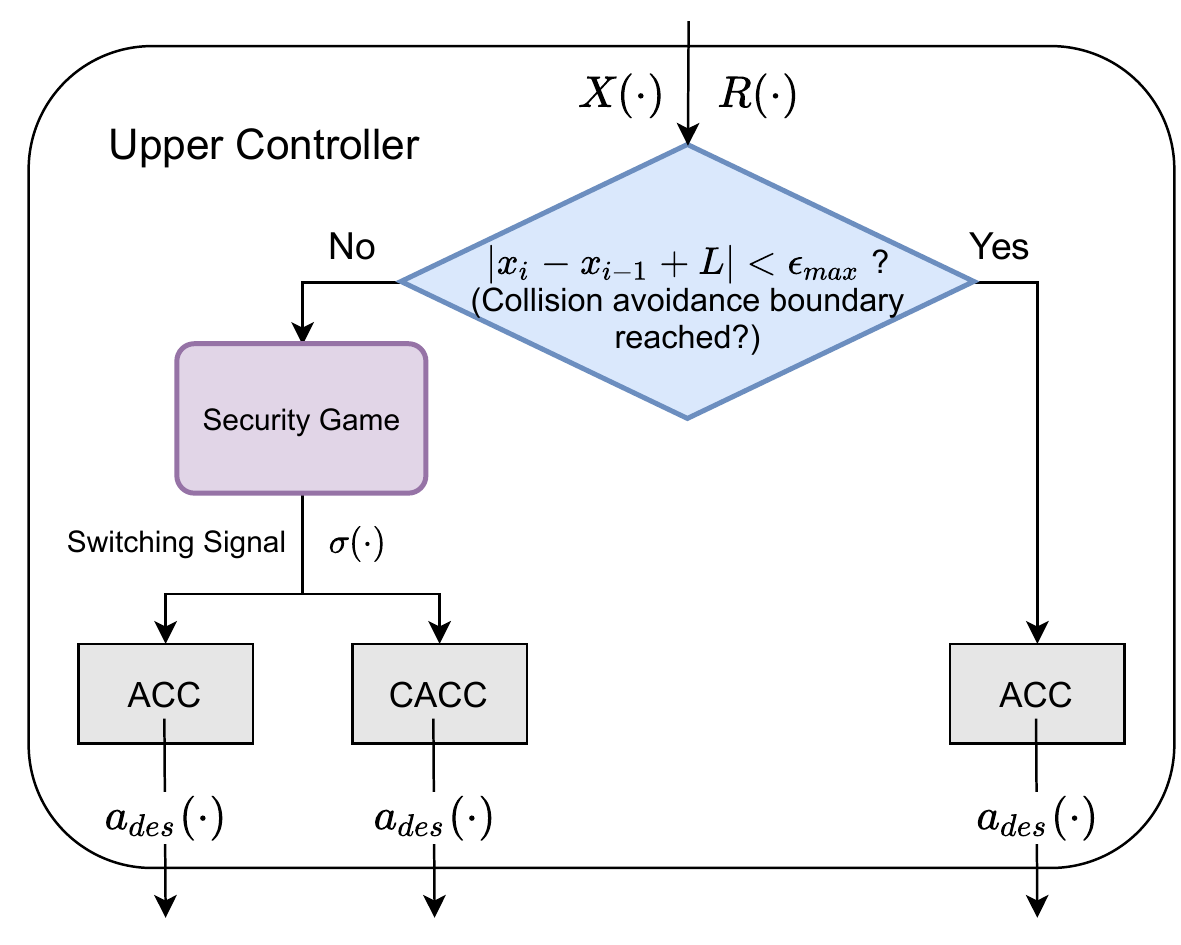}
    \caption{Structure of Upper Controller.}
    \label{fig:UpperControllerStructure}
\end{figure}

 \subsection{A Security Game-based Switched System}

Given the benefits and limitations of both ACC and CACC controllers discussed in Section~\ref{Attack Model}, we suggest using ACC as a secondary controller operating as a back-up when the communication network is suspicious. In this way, the advantages of both controllers will be retained while minimising the effects caused by cyber-attacks. 

\begin{figure}[tb]
    \centering
    \includegraphics[width=0.8\columnwidth]{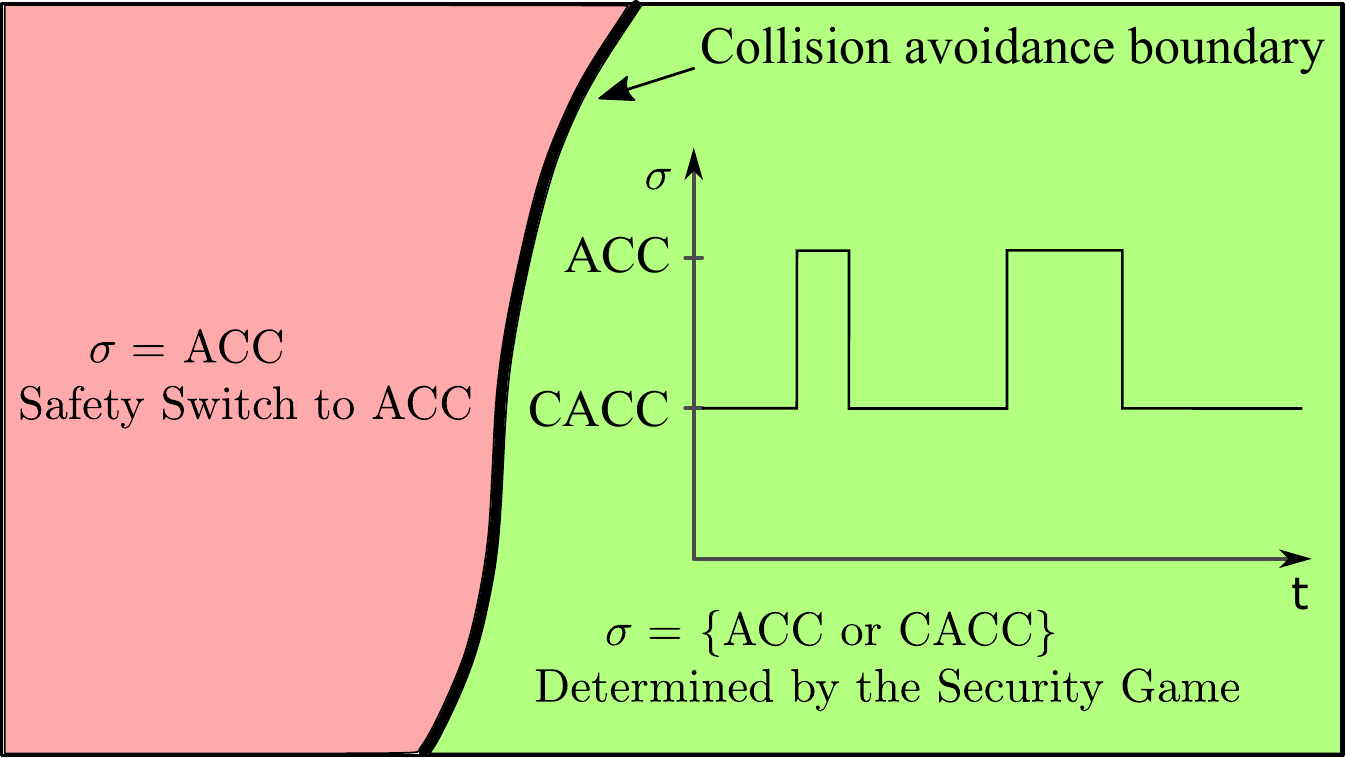}
    \caption{State space representation of the switched system.} 
    \label{fig:switchedspace}
\end{figure}
The overall structure of the improved upper controller is shown in Fig.~\ref{fig:UpperControllerStructure}, which delivers insight into attack detection and mitigation. Both controllers form a switched system with the switching decision coming from novel game-theoretic analysis of the attacker and the defender's interactions. Details of the game structure will be discussed in Section~\ref{Game}. A dedicated state constraint further enhances safety against bounded but aggressive message modifications in which the solution is bounded but violates practical constraints ($\epsilon_{max}$), which represent vehicles nearly crashing. The resulting state space is visualised in Fig.~\ref{fig:switchedspace}.

Formally, the switched system combined attack disturbance can be written in matrix form
\begin{subequations}\label{switchedSys}
\begin{equation}
\begin{split}
    \dot z = A_p z  +B_pR+M_p\mathbf{\xi}, \ \ p\in \mathcal{P},
    \end{split}
\end{equation}
\begin{equation}
    p = \begin{cases}
      ACC, & \text{if $|x_i-x_{i-1}+L| < \epsilon_{max}$},\\
      \sigma(t),& \text{otherwise},
    \end{cases}  
\end{equation}
\end{subequations}
where $\mathcal{P}$ is an index set of available subsystems and $\mathcal{P} = \{CACC, ACC\}$ consisting of two control systems in our case, $\sigma(t): [0,\infty) \rightarrow \mathcal{P}$ is a piecewise right-continuous constant function generated based on the security game that specifies the index of the active system at each time, $\mathbf{\xi}$ represent attack effects, $\{(A_p,B_p,M_p)\}_{p\in \mathcal{P}}$ is a set of state matrix triples with suitable dimensions for different control systems and the state constraint $|x_i-x_{i-1}+L| < \epsilon_{max}$ is the aforementioned switching surface, which defines the collision avoidance boundary in Fig.~\ref{fig:UpperControllerStructure} and \ref{fig:switchedspace}. Note that, $M_{ACC}=0$ based on Assumption~\ref{def:CommAttacknotAffectSensor}.

\subsection{Stability Analysis}
Stability analysis of such a switched system is necessary: even if both controllers satisfy the individual vehicle stability condition~\eqref{indStableCond}, unconstrained switching may destabilize such a switched system \cite{liberzon2003switching}.

\begin{definition}[Global Uniform Exponential Stability  \cite{liberzon2003switching}]
A switched system has global uniform exponential stability (GUES) if there exist a positive constant $\delta>0$ such that for all switching signals $\sigma$ the solutions of (\ref{switchedSys}) with initial state $|z(0)|\leq \delta$ satisfy 
\begin{equation}\label{switchedSys_stableCond}
    |z(t)|\leq c|z(0)|\exp^{-\lambda t}, \ \ \forall t\geq 0, 
\end{equation} for some $c,\lambda > 0$.
\end{definition}

\begin{theorem}
If conditions (\ref{commonV}) are satisfied, then all sub-systems of the platoon in (\ref{switchedSys}) share a radially unbounded common Lyapunov function, and therefore the switched system has global uniform exponential stability.
\end{theorem}

\begin{proof}
The proof is based on the results in Chapter 2 of \cite{liberzon2003switching}. It is natural to consider quadratic common Lyapunov functions of the form \eqref{quadraticLya} for switched linear systems, 
 \begin{equation}\label{quadraticLya}
    V(z) = z^T P z, \ \ \ \ P=P^T>0.
\end{equation}
Assuming that \eqref{CACCinvCod} and \eqref{ACCinvCod} are satisfied, we have to find a positive definite symmetric matrix $P$ such that the inequality \eqref{Lyaineqs} is fulfilled
\begin{equation}\label{Lyaineqs}
    A_p^TP+PA_p<0,\ \ \ \forall p\in \mathcal{P}.
\end{equation}
The inequality $M < N$ for two symmetric matrices M and N means that the matrix $M-N$ is negative definite. Note that, if a matrix is positive (negative) definite, all its eigenvalues are positive (negative).

For the switched system consisting of $A_{CACC}$ \eqref{A:CACC} and $A_{ACC}$ \eqref{A:ACC}, $P=\left[
\begin{array}{cc}
 p_{11} & p_{12} \\
 p_{12} & p_{22} \\
\end{array}
\right]$ has to satisfy \eqref{commonV} to guarantee GUES.

\begin{subequations}\label{commonV}
\begin{align}
    p_{11}>0,  & \;\; p_{12}>0, \\
    p_{22}&>\frac{p_{12}^2}{p_{11}},\\
    k_1<0,  & \;\; k_{3}<0,
\end{align}
\\[-2.75em]
\begin{align}
    k_2&>\frac{-2 p_{12} \sqrt{k_1-\frac{k_{1} p_{11} p_{22}}{p_{12}^2}}+k_{1} p_{22}-p_{11}}{p_{12}},\\
        k_{2}&<\frac{2 p_{12} \sqrt{k_{1}-\frac{k_{1} p_{11} p_{22}}{p_{12}^2}}+k_{1} p_{22}-p_{11}}{p_{12}},\\
          k_{4}&>\frac{-2 p_{12} \sqrt{k_{3}-\frac{k_{3} p_{11} p_{22}}{p_{12}^2}}+k_{3} p_{22}-p_{11}}{p_{12}},\\
     k_{4}&<\frac{2 p_{12} \sqrt{k_{3}-\frac{k_{3} p_{11} p_{22}}{p_{12}^2}}+k_{3} p_{22}-p_{11}}{p_{12}}.
\end{align}
     
\end{subequations}
\end{proof}

\begin{remark}
The conditions in \eqref{commonV} are sufficient but not necessary to achieve GUES. There may be other types of common Lyapunov functions other than \eqref{quadraticLya}, which may lead to other sufficient conditions.
\end{remark}

String stability of vehicular platooning is important for long-term stable platoon operation. However, as discussed in Section~\ref{BasicStability}, sensor-based control policy ACC fails to guarantee this property. Due to erroneous detection results (e.g., false alarms), the switching signal $\sigma(\cdot)$ generated from the security game may initiate control system reconfiguration process even if the vehicle is not exposed to an attack. Therefore, we present a lower bound on the dwell time $\tau_n$ for the string stable controller CACC to retain platoon string stability of the switched system in an attack-free environment. This constraint is updated dynamically at the beginning of each interval on which $\sigma=CACC$ based on the current system states.
\begin{theorem}\label{StringinBenignEnv}
Consider a dual-mode control system reconfiguration scheme that consists of a control policy (e.g., CACC) which guarantees platoon string stability and a control policy (e.g., ACC) which cannot guarantee string stability. Assume CACC is globally exponentially string stable with a Lyapunov function $V$ satisfying 
\begin{equation}\label{eq:radiallyUnbounded}
    a\abs{z}^2\leq V(z)\leq b\abs{z}^2,
\end{equation}
and \begin{equation}\label{eq:negDefinite}
    \dot V(z)\leq -c\abs{z}^2,
\end{equation}
for some positive constants $a$, $b$ and $c$.
Suppose switching signal $\sigma$ alternatively chooses CACC on $[t_n, t_{n+1})$ and ACC on $[t_{n+1},t_{n+2})$ where $n\in\mathbb{Z}_{\geq 0}$, and repeats this infinitely many times. The switched system \eqref{switchedSys} guarantees platoon string stability if the dwell time $\tau_n$ for the stable system CACC satisfies
\begin{equation}\label{eq:switchingConstraint}
    \tau_n > \frac{1}{\lambda}\log \abs{z(t_n)},
\end{equation}
where $\lambda = \frac{c}{2b}$.
\end{theorem}
\begin{proof}
Combining \eqref{eq:radiallyUnbounded} and \eqref{eq:negDefinite}, we have
\begin{equation*}
    \dot V(z)\leq -2\lambda V(z),
\end{equation*} where $\lambda = \frac{c}{2b}$.
This leads to the inequality
\begin{equation}\label{eq:LyaChangedForm}
    V(z(t_n+\tau_n))\leq e^{-2\lambda \tau}V(z(t_n)).
\end{equation} provided that $\sigma(t) = CACC$ for $t\in [t_n, t_n+\tau_n)$,

\noindent 
Suppose $\delta =  \|\epsilon_{i-1}\|_\infty -\|\epsilon_i\|_\infty$ is the worst-case disturbance amplification as disturbance propagates down to the tail of the platoon with $\|\epsilon_i\|_\infty$ defined by \eqref{def:stringStablity}.
To compute an explicit lower bound on $\tau_n$ to guarantee string stability for the switched system, it is sufficient to ensure that
\begin{equation}\label{eq:essenseLyaCond}
    V(z(t_n))>V(z(t_{n+2})),
\end{equation}

Note that, $V(z(t_n+\tau_n)) = V(z(t_{n+1}))$ and $V(z(t_{n+1})+\delta) = V(z(t_{n+2}))$ because the Lyapunov function is continuous.
Therefore, it is equivalent to satisfy 
\begin{equation}\label{eq:interStep1}
    V(z(t_n))-V(z(t_n+\tau_n))> V(z(t_{n+1})+\delta)-V(z(t_{n+1})).
\end{equation}
From \eqref{eq:LyaChangedForm}, we have 
\begin{equation*}
    V(z(t_n))-V(z(t_n+\tau_n)) \geq V(z(t_n))-e^{-2\lambda \tau_n}V(z(t_n)).
\end{equation*}
Then, \eqref{eq:interStep1} will hold, if we have
\begin{equation*}
    V(z(t_n))+V(z(t_{n+1}))>e^{-2\lambda \tau_n}V(z(t_n))+V(x(t_{n+1})+\delta).
\end{equation*}
By virtue of \eqref{eq:radiallyUnbounded}, we have
\begin{subequations}
\begin{equation*}
    b\abs{z(t_n)}^2+b\abs{z(t_{n+1})}^2 \geq V(z(t_n))+V(z(t_{n+1})).
\end{equation*}
    \begin{equation*} 
       \scalebox{0.9}{$e^{-2\lambda \tau_n}V(z(t_n))+V(z(t_{n+1})+\delta) \geq a e^{-2\lambda \tau_n} \abs{z(t_n)}^2+a\abs{z(t_{n+1})+\delta}^2$}.
    \end{equation*}
\end{subequations}
Then, all we need to have is 
\begin{equation*}
    b\abs{z(t_n)}^2+b\abs{z(t_{n+1})}^2 >a e^{-2\lambda \tau_n} \abs{z(t_n)}^2,
\end{equation*} which can be rewritten as 
\begin{equation*}
    \tau_n > \frac{1}{2\lambda}\log{\frac{a\abs{z(t_n)}^2}{b(\abs{z(t_{n+1})}^2+\abs{z(t_n)}^2)}}.
\end{equation*}
This immediately yields the following lower bound on dwell time in
CACC:
\begin{equation*}
     \tau_n > \frac{1}{\lambda}\log \abs{z(t_n)}.
\end{equation*}
\end{proof}

Note that the assumption on exponential stability may be justified
in certain cases as discussed in \cite{feng2019string}. Other estimates in \eqref{eq:radiallyUnbounded} and \eqref{eq:negDefinite} can be used instead of quadratic ones. In essence, all we need is for \eqref{eq:essenseLyaCond} to hold for all switching times. We observe the further the system state deviation from the equilibrium point the larger $\tau_n$ should be, meaning that CACC controller should be activated for a longer time to compensate the negative effects brought by ACC in terms of string stability. Moreover, if CACC could be tuned to converge faster, i.e. $c$ is large, then
$\tau_n$ can be smaller.

\begin{remark}
The switched system may not guarantee string stability when the communication-based controller CACC is compromised in an adversarial environment. In this case, the proposed switched system can be seen as an emergency measure to prevent a collision. Providing a switching function that also guarantees string stability for this case is a research direction for future work.
\end{remark}

\subsection{Numerical Example}
  If state matrices \eqref{A:CACC} and \eqref{A:ACC} of two control systems take the following control gains
\begin{subequations}
 \begin{equation*}
         A_{CACC} = \left[
\begin{array}{cc}
 0 & 1 \\
 k_1 & k_2 \\
\end{array}
\right] = \left[
\begin{array}{cc}
 0 & 1 \\
 -1.58 & -2.51 \\
\end{array}
\right],
\end{equation*}

 \begin{equation*}
  A_{ACC\ } = \left[
\begin{array}{cc}
 0 & 1 \\
 k_3 & k_4 \\
\end{array}
\right] = \left[
\begin{array}{cc}
 0 & 1 \\
 -0.25 & -1 \\
\end{array}
\right]. \end{equation*}

\end{subequations}
A suitable symmetric positive definite matrix that satisfies \eqref{Lyaineqs} is 
\begin{equation*}
        P = \left[
\begin{array}{cc}
 1. & 0.154297 \\
 0.154297 & 1.57813 \\
\end{array}
\right].
\end{equation*}

\section{SECURITY GAME FORMULATION} \label{Game}

We propose a non-cooperative cybersecurity defense game played between the \textit{Attacker}, the anomaly detector and the \textit{Defender} to guide the controller switching process in an online fashion. By \textit{Attacker} and \textit{Defender} we mean the malicious adversaries and the unit that generates switching signals respectively. Note that, both players' actions and detection reports are represented as edges and the resulting states are represented as nodes in the game tree as in Fig.~\ref{fig:extensiveGame}. The game begins with the Attacker choosing whether to attack the vehicle platoon or not, represented by the leftmost branches in red. If the Attacker chooses to attack then it performs the \textit{message falsification} attack as modelled by \eqref{eq:CACC_withAttack}. 
The notion of ``detector'' is kept very general to make sure our game structure is suitable to various types anomaly detection approaches with the understanding that the design process for all detectors can not perfectly anticipate all complex real-world situations and attack models leading to a certain probability of error.
\begin{definition}[Chance node]
A chance node can be seen as a fictitious player who performs actions according to a probability distribution. 
\end{definition}
We use chance nodes to model the uncertainty about detection results as highlighted in green in Fig.~\ref{fig:extensiveGame}. There should be an ongoing update process regarding these prior beliefs based on detection results in the real-world deployment. 
\begin{definition}[Information set]
An information set of player $i$ is a collection of player $i$'s nodes among which $i$ cannot distinguish. 
\end{definition}
Once the Defender obtains the detection results, it is unclear whether an actual attack has been performed or not. This unique situation is modeled by \textit{information sets} shown as dashed lines in the figure. There are two information sets for the Defender based on the detection reports: one indicating an attack and the other for no attack. This means the Defender must consider the consequences of both an actual attack having occurred, and no attack having occurred, when an attack has been reported by the detector. 

Lastly, after considering a rational Attacker's action and the chances of detection errors, the Defender decides whether to downgrade the CACC controller to the ACC controller or remain with the CACC controller, for example, in the case when an attack is reported with low but non-zero probability. The dual-mode control system reconfiguration scheme leverages the Defender's decisions in the game as the switching signal $\sigma$ of the switched control system. As shown in green in Fig.~\ref{fig:switchedspace}, the solution of the game activates one of the subsystems (i.e. ACC or CACC) which in turn generates new state values for another game before the next switching is required.
 
Formally, we model the game with
\begin{itemize}
    \item Attacker's action space $\mathpzc{A}^A:=\ $\{$a$: engage an attack; $na$: not attacking\}
    \item Chance nodes (anomaly detector) $C:=\ $\{$r$: reporting an attack; $nr$: not reporting an attack\}
    \item Defender's action space $\mathpzc{A}^D:=\ $\{$d$: switch to ACC; $nd$: switch to CACC\}
\end{itemize}
The strategy profile is modeled as $\langle a, c, d \rangle$ for $a\in A$, $c\in C$ and $d\in D$.
The utility values for the Attacker and the Defender are denoted as $[(R_1^A, R_1^D), \dots, (R_8^A, R_8^D)]$, which can be chosen to reflect specific vehicle platoon security trade-offs and risks.

\subsection{Numerical Example} 
\begin{figure}[bp]
    \centering
    \includegraphics[width=\linewidth]{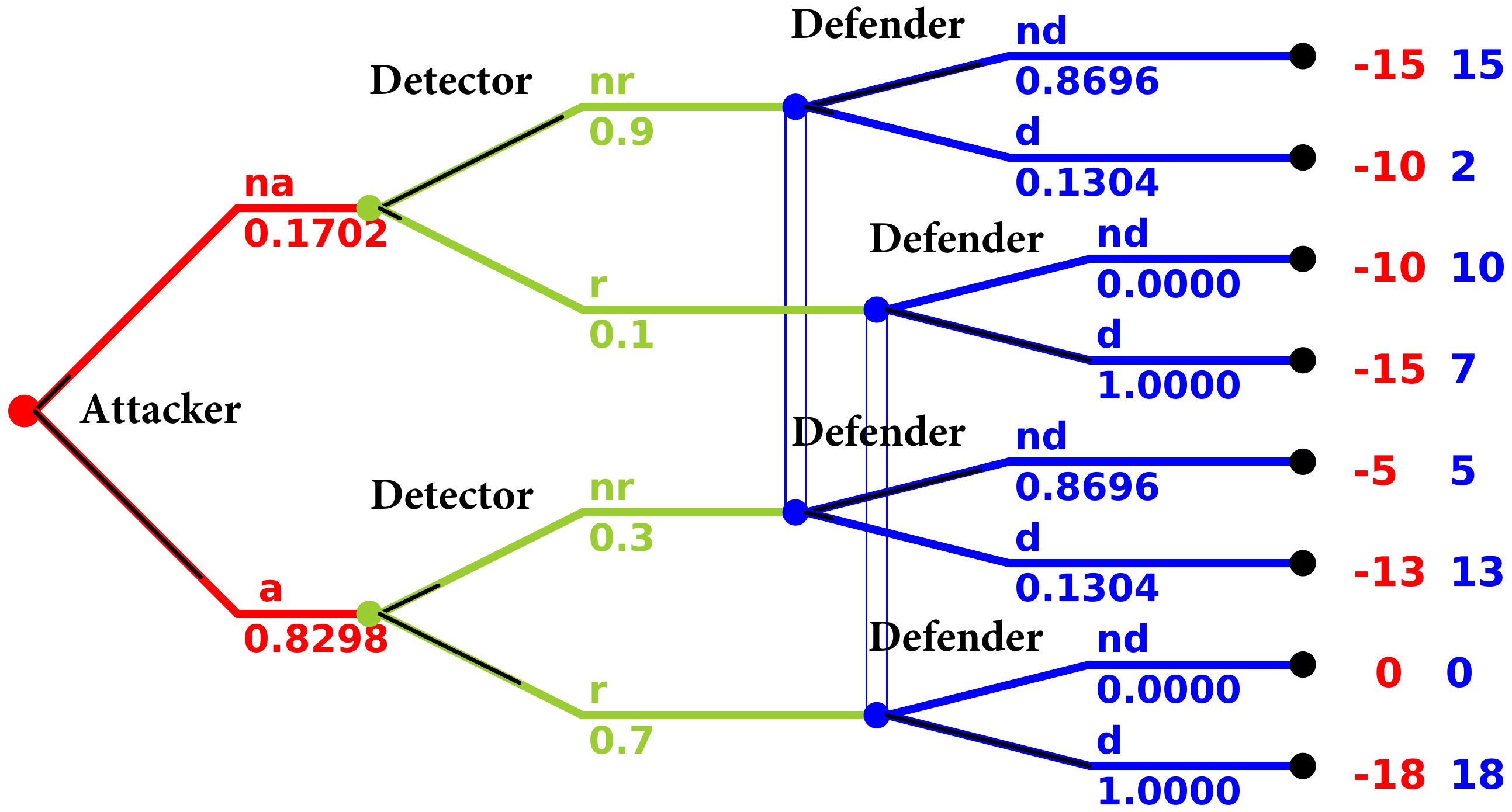}
    \caption{An example game model: the attacker's actions are in red, detection results are in green, and defender's controller switching decisions are in blue.}
    \label{fig:extensiveGame}
\end{figure}
One instance of the game is shown graphically as in Fig.~\ref{fig:extensiveGame}. The utilities of each strategy profile are highlighted in red for the Attacker and blue for the Defender at the leaves of the game tree. The example anomaly detector: $90\%$ of the time correctly reports benign data when there is no attack; it makes false alarms for the remaining $10\%$ of the time. When an attack  has truly occurred, this detector correctly reports it $70\%$  of the time and misses for the remaining $30\%$ of the time. Note that, these values should be calibrated based on the deployed anomaly detector in different problem settings. 
We use a popular open-source game solver \textit{Gambit} (\cite{mckelvey2006gambit}) to find Nash equilibrium solutions. A unique mixed strategy is numerically computed. The probability distribution of each player's actions is shown under the figure's edges.  For example, the Attacker with the above utilities will attack with $82.98\%$ of a chance to attack and if the Detector reports no attack being detected, the Defender would still choose to downgrade to sensor based ACC controller with $13.04\%$ of the chance to react to this high attack intention and imperfect detection results. 
\begin{figure}[t]
     \centering
     \begin{subfigure}[b]{0.45\linewidth}
         \centering
         \includegraphics[width=\linewidth]{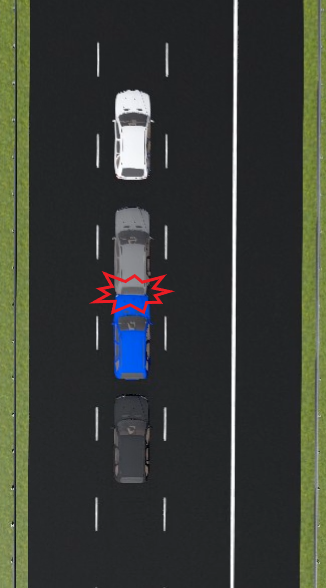}
         \caption{Without attack mitigation}
         \label{fig:noD}
     \end{subfigure}
     \hfill
     \begin{subfigure}[b]{0.45\linewidth}
         \centering
         \includegraphics[width=\linewidth]{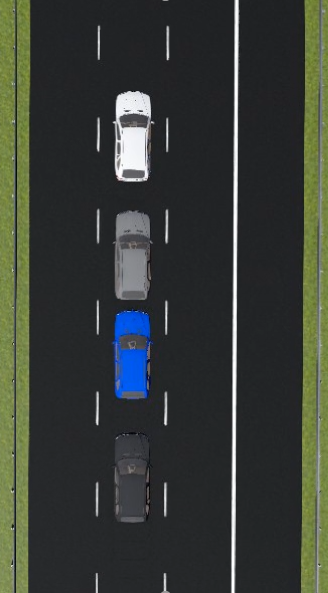}
         \caption{With attack mitigation}
         \label{fig:withD}
     \end{subfigure}
    \hfill
    \caption{Simulation results of a vehicle platoon of size $4$ under message falsification attack, showing a crash without attack mitigation (a) and crash prevention (b).}
    \label{fig:simulation}
\end{figure}

\subsection{Simulation Example}
The effectiveness of the proposed dual-mode control system reconfiguration scheme can be further demonstrated by comparing simulation results as in Fig.~\ref{fig:simulation}. The simulations are conducted in a sophisticated simulator \textit{Webots} \cite{Webots}, which supports realistic simulation of traffic flow, the ambient environmental impact and the vehicle's physical properties (e.g., motor torques, body mass, suspension etc.). In the simulation, the communication channel of the blue \textit{BMW X5} vehicle is compromised. If there is no attack mitigation, malicious message modification would cause the vehicle to accelerate and collide with its predecessor as shown in Fig.~\ref{fig:noD}. With the proposed approach, the attack effects could be significantly reduced. Instead of a collision, the adversary could at most cause the vehicle to reduce the inter-vehicle distance. After reaching the minimum possible inter-vehicle distance as seen in Fig.~\ref{fig:withD}, the scheme is able to guide the vehicle to regain its desired inter-vehicle distance and relative velocity. 

\begin{remark}
The Nash equilibrium solution is highly dependent on the utility values of both players and the detection probabilities of the detector. Therefore, it is critical that these values properly reflect the trade-offs and risks in terms of platoon security. A detailed design of a machine learning based  anomaly detector and a set of utility functions that fulfills these requirements is ongoing research.  
\end{remark}
\section{CONCLUSION}\label{conclusion}

Although V2V communication empowers vehicle platoons with improved operation stability, the resulting high level of connectivity and openness may attract  communications-based cyber-physical attacks. Consequently, we have investigated a controller reconfiguration scheme to mitigate attack effects, and thereby enhanced system safety in an adversarial environment. This paper constitutes a first attempt to use security games to guide the switching process in the context of switched systems, where the interactions between an intelligent attacker and defender in possession of an imperfect detector have been investigated. Two common controllers CACC and ACC for autonomous vehicle platooning have been carefully analysed. A sliding surface based on a prefixed state constraint highlights the inadequacy of common stability definitions
in this context and further guarantees  system safety. A minimum dwell time constraint is derived to ensure string stability of the switched system in a benign environment. While the presented approach
is effective for attack mitigation in short-term operation, we cannot ensure string stability in an adversarial environment, which would possibly require
novel control designs beyond CACC or ACC.
Hence, an interesting open problem motivated by our work is to ascertain whether a switching signal could achieve string stability under communication-based attacks (e.g., to find probabilistic string stability guarantees based on the characteristics of the imperfect anomaly detector). Moreover, additional simulation studies could support usefulness and reliability when the low-level controller and vehicle models are taken into account. 

\section{ACKNOWLEDGMENTS}
We gratefully acknowledge support from the DSTG Next Generation Technology Fund and CSIRO Data61 CRP on `Adversarial Machine Learning for Cyber', and CSIRO Data61 PhD scholarship.

\bibliographystyle{ieeetr}
\bibliography{CDC}

\begin{thebibliography}{10}

\bibitem{qayyum2020securing}
A.~Qayyum, M.~Usama, J.~Qadir, and A.~Al-Fuqaha, ``Securing connected \&
  autonomous vehicles: Challenges posed by adversarial machine learning and the
  way forward,'' {\em IEEE Communications Surveys \& Tutorials}, vol.~22,
  no.~2, pp.~998--1026, 2020.

\bibitem{boeira2017effects}
F.~Boeira, M.~P. Barcellos, E.~P. de~Freitas, A.~Vinel, and M.~Asplund,
  ``Effects of colluding sybil nodes in message falsification attacks for
  vehicular platooning,'' in {\em 2017 IEEE Vehicular Networking Conference
  (VNC)}, pp.~53--60, IEEE, 2017.

\bibitem{sumra2015attacks}
I.~A. Sumra, H.~B. Hasbullah, and J.-l.~B. AbManan, ``Attacks on security goals
  (confidentiality, integrity, availability) in vanet: a survey,'' in {\em
  Vehicular Ad-Hoc Networks for Smart Cities}, pp.~51--61, Springer, 2015.

\bibitem{wiedersheim2010privacy}
B.~Wiedersheim, Z.~Ma, F.~Kargl, and P.~Papadimitratos, ``Privacy in
  inter-vehicular networks: Why simple pseudonym change is not enough,'' in
  {\em 2010 Seventh international conference on wireless on-demand network
  systems and services (WONS)}, pp.~176--183, IEEE, 2010.

\bibitem{zhang2020distributed}
D.~Zhang, Y.-P. Shen, S.-Q. Zhou, X.-W. Dong, and L.~Yu, ``Distributed secure
  platoon control of connected vehicles subject to dos attack: theory and
  application,'' {\em IEEE Transactions on Systems, Man, and Cybernetics:
  Systems}, 2020.

\bibitem{feng2020dynamic}
S.~Feng and H.~Ishii, ``Dynamic quantized leader-follower consensus under
  denial-of-service attacks,'' in {\em 2020 59th IEEE Conference on Decision
  and Control (CDC)}, pp.~488--493, IEEE, 2020.

\bibitem{merco2018replay}
R.~Merco, Z.~A. Biron, and P.~Pisu, ``Replay attack detection in a platoon of
  connected vehicles with cooperative adaptive cruise control,'' in {\em 2018
  Annual American Control Conference (ACC)}, pp.~5582--5587, IEEE, 2018.

\bibitem{keijzer2019sliding}
T.~Keijzer and R.~M. Ferrari, ``A sliding mode observer approach for attack
  detection and estimation in autonomous vehicle platoons using event triggered
  communication,'' in {\em 2019 IEEE 58th Conference on Decision and Control
  (CDC)}, pp.~5742--5747, IEEE, 2019.

\bibitem{dadras2018identification}
S.~Dadras, S.~Dadras, and C.~Winstead, ``Identification of the attacker in
  cyber-physical systems with an application to vehicular platooning in
  adversarial environment,'' in {\em 2018 Annual American Control Conference
  (ACC)}, pp.~5560--5567, IEEE, 2018.

\bibitem{yang2019tree}
L.~Yang, A.~Moubayed, I.~Hamieh, and A.~Shami, ``Tree-based intelligent
  intrusion detection system in internet of vehicles,'' in {\em 2019 IEEE
  Global Communications Conference (GLOBECOM)}, pp.~1--6, IEEE, 2019.

\bibitem{alotibi2020anomaly}
F.~Alotibi and M.~Abdelhakim, ``Anomaly detection for cooperative adaptive
  cruise control in autonomous vehicles using statistical learning and
  kinematic model,'' {\em IEEE Transactions on Intelligent Transportation
  Systems}, 2020.

\bibitem{sunstrategic}
G.~Sun, T.~Alpcan, B.~I.~P. Rubinstein, and S.~Camtepe, ``Strategic mitigation
  against wireless attacks on autonomous platoons,'' in {\em Machine Learning
  and Knowledge Discovery in Databases: Applied Data Science Track}, ECML-PKDD,
  pp.~69--84, 2021.

\bibitem{alpcan2010network}
T.~Alpcan and T.~Ba{\c{s}}ar, {\em Network security: A decision and
  game-theoretic approach}.
\newblock Cambridge University Press, 2010.

\bibitem{sedjelmaci2016lightweight}
H.~Sedjelmaci, S.~M. Senouci, and M.~Al-Bahri, ``A lightweight anomaly
  detection technique for low-resource iot devices: A game-theoretic
  methodology,'' in {\em 2016 IEEE international conference on communications
  (ICC)}, pp.~1--6, IEEE, 2016.

\bibitem{fang2016deploying}
F.~Fang, T.~H. Nguyen, R.~Pickles, W.~Y. Lam, G.~R. Clements, B.~An, A.~Singh,
  M.~Tambe, A.~Lemieux, {\em et~al.}, ``Deploying paws: Field optimization of
  the protection assistant for wildlife security.,'' in {\em AAAI}, vol.~16,
  pp.~3966--3973, 2016.

\bibitem{subba2018game1}
B.~Subba, S.~Biswas, and S.~Karmakar, ``A game theory based multi layered
  intrusion detection framework for wireless sensor networks,'' {\em
  International Journal of Wireless Information Networks}, vol.~25, no.~4,
  pp.~399--421, 2018.

\bibitem{zohdy2012game}
I.~H. Zohdy and H.~Rakha, ``Game theory algorithm for intersection-based
  cooperative adaptive cruise control (cacc) systems,'' in {\em 2012 15th
  International IEEE Conference on Intelligent Transportation Systems},
  pp.~1097--1102, IEEE, 2012.

\bibitem{dextreit2013game}
C.~Dextreit and I.~V. Kolmanovsky, ``Game theory controller for hybrid electric
  vehicles,'' {\em IEEE Transactions on Control Systems Technology}, vol.~22,
  no.~2, pp.~652--663, 2013.

\bibitem{marden2015game}
J.~R. Marden and J.~S. Shamma, ``Game theory and distributed control,'' in {\em
  Handbook of game theory with economic applications}, vol.~4, pp.~861--899,
  Elsevier, 2015.

\bibitem{junhui2013power}
Z.~Junhui, Y.~Tao, G.~Yi, W.~Jiao, and F.~Lei, ``Power control algorithm of
  cognitive radio based on non-cooperative game theory,'' {\em China
  Communications}, vol.~10, no.~11, pp.~143--154, 2013.

\bibitem{yang2017simultaneous}
J.~Yang, Y.~Chen, F.~Zhu, and F.~Wang, ``Simultaneous state and output
  disturbance estimations for a class of switched linear systems with unknown
  inputs,'' {\em International Journal of Systems Science}, vol.~48, no.~1,
  pp.~22--33, 2017.

\bibitem{zammali2020interval}
C.~Zammali, J.~Van~Gorp, and T.~Ra{\"\i}ssi, ``Interval observers based fault
  detection for switched systems with $l_\infty$ performances,'' in {\em 2020
  European Control Conference (ECC)}, pp.~1053--1056, IEEE, 2020.

\bibitem{van2014hybrid}
J.~Van~Gorp, M.~Defoort, K.~C. Veluvolu, and M.~Djemai, ``Hybrid sliding mode
  observer for switched linear systems with unknown inputs,'' {\em Journal of
  the Franklin Institute}, vol.~351, no.~7, pp.~3987--4008, 2014.

\bibitem{lee2006optimal}
J.-W. Lee and G.~E. Dullerud, ``Optimal disturbance attenuation for
  discrete-time switched and markovian jump linear systems,'' {\em SIAM Journal
  on Control and Optimization}, vol.~45, no.~4, pp.~1329--1358, 2006.

\bibitem{yang2016finite}
G.~Yang and D.~Liberzon, ``Finite data-rate stabilization of a switched linear
  system with unknown disturbance,'' {\em IFAC-PapersOnLine}, vol.~49, no.~18,
  pp.~1085--1090, 2016.

\bibitem{sanchez2019practical}
C.~A. Sanchez, G.~Garcia, S.~Hadjeras, W.~Heemels, and L.~Zaccarian,
  ``Practical stabilization of switched affine systems with dwell-time
  guarantees,'' {\em IEEE Transactions on Automatic Control}, vol.~64, no.~11,
  pp.~4811--4817, 2019.

\bibitem{rajamani2011vehicle}
R.~Rajamani, {\em Vehicle dynamics and control}.
\newblock Springer Science \& Business Media, 2011.

\bibitem{segata2014plexe}
M.~Segata, S.~Joerer, B.~Bloessl, C.~Sommer, F.~Dressler, and R.~L. Cigno,
  ``Plexe: A platooning extension for veins,'' in {\em 2014 IEEE Vehicular
  Networking Conference (VNC)}, pp.~53--60, IEEE, 2014.

\bibitem{zheng2015stability}
Y.~Zheng, S.~E. Li, J.~Wang, D.~Cao, and K.~Li, ``Stability and scalability of
  homogeneous vehicular platoon: Study on the influence of information flow
  topologies,'' {\em IEEE Transactions on intelligent transportation systems},
  vol.~17, no.~1, pp.~14--26, 2015.

\bibitem{chen1984linear}
C.-T. Chen and C.-T. Chen, {\em Linear system theory and design}, vol.~301.
\newblock Holt, Rinehart and Winston New York, 1984.

\bibitem{liberzon2003switching}
D.~Liberzon, {\em Switching in systems and control}.
\newblock Springer Science \& Business Media, 2003.

\bibitem{feng2019string}
S.~Feng, Y.~Zhang, S.~E. Li, Z.~Cao, H.~X. Liu, and L.~Li, ``String stability
  for vehicular platoon control: Definitions and analysis methods,'' {\em
  Annual Reviews in Control}, vol.~47, pp.~81--97, 2019.

\bibitem{mckelvey2006gambit}
R.~D. McKelvey, A.~M. McLennan, and T.~L. Turocy, ``Gambit: Software tools for
  game theory,'' 2006.
\newblock Version 0.2006. 01.20.

\bibitem{Webots}
Webots, ``http://www.cyberbotics.com.''
\newblock Open-source Mobile Robot Simulation Software.

\end{thebibliography}

\end{document}